\documentclass[conference]{IEEEtran}
\IEEEoverridecommandlockouts
\usepackage{cite}
\usepackage{amsmath,amssymb,amsfonts,amsthm}
\usepackage{algorithmic}
\usepackage{graphicx}
\usepackage{textcomp}
\usepackage{xcolor}
\usepackage{bm}
\usepackage{enumerate}
\usepackage{tabu}
\usepackage{makecell}
\usepackage{multirow}
\usepackage{etoolbox}
\makeatletter
\patchcmd{\@makecaption}
  {\scshape}
  {}
  {}
  {}
\makeatother
\theoremstyle{definition}
\newtheorem{theorem}{Theorem}
\newtheorem{proposition}{Proposition}
\newtheorem{corollary}{Corollary}
\newtheorem{definition}{Definition}

\newtheorem{example}{Example}
\newtheorem*{remark}{Remark}
\def\BibTeX{{\rm B\kern-.05em{\sc i\kern-.025em b}\kern-.08em
    T\kern-.1667em\lower.7ex\hbox{E}\kern-.125emX}}
\begin{document}

\title{Steering control of payoff-maximizing players in adaptive learning dynamics
\thanks{X.C. gratefully acknowledges the generous faculty startup fund provided by BUPT (No. 505022023).}
}

\author{\IEEEauthorblockN{1\textsuperscript{st} Xingru Chen}
\IEEEauthorblockA{\textit{School of Sciences} \\
\textit{Beijing University of Posts and Telecommunications}\\
Beijing 100876, China \\
xingrucz@gmail.com}
\and
\IEEEauthorblockN{2\textsuperscript{nd} Feng Fu}
\IEEEauthorblockA{\textit{Department of Mathematics} \\
\textit{Department of Biomedical Data Science} \\
\textit{Dartmouth College}\\
Hanover, NH 03755, USA \\
fufeng@gmail.com}
}
\maketitle

\begin{abstract}
Evolutionary game theory provides a mathematical foundation for cross-disciplinary fertilization, especially for integrating ideas from artificial intelligence and game theory. Such integration offers a transparent and rigorous approach to complex decision-making problems in a variety of important contexts, ranging from evolutionary computation to machine behavior.  Despite the astronomically huge individual behavioral strategy space for interactions in the iterated Prisoner's Dilemma (IPD) games, the so-called Zero-Determinant (ZD) strategies is a set of rather simple memory-one strategies yet can unilaterally set a linear payoff relationship between themselves and their opponent. Although the witting of ZD strategies gives players an upper hand in the IPD games, we find and characterize unbending strategies that can force ZD players to be fair in their own interest. Moreover, our analysis reveals the ubiquity of unbending properties in common IPD strategies which are previously overlooked. In this work, we demonstrate the important steering role of unbending strategies in fostering fairness and cooperation in pairwise interactions. Our results will help bring a new perspective by means of combining game theory and multi-agent learning systems for optimizing winning strategies that are robust to noises, errors, and deceptions in non-zero-sum games.
\end{abstract}


\section{Introduction}

Evolutionary game theory provides a mathematical foundation for studying mechanisms of cooperation and myriad learning theories toward altruistic behavior~\cite{axelrod1981evolution, nowak2006five}. Adaptive learning strategies will evolve in scenarios where self-interested individuals interact with one another. Among others, the Prisoner's Dilemma (PD) game is a symmetric game involving two players X and Y, and two actions: to cooperate or to defect. In a one-shot PD game, the four possible outcomes correspond to different payoffs from the focal player's perspective: if both are cooperators, one gets the reward $R$, if a cooperator is against a defector, the sucker's payoff $S$, if a defector is against a cooperator, the temptation $T$, and if both are defectors, the punishment $P$. The game is considered a paradigm for understanding the conflict between self-interest and collective interest as the payoff structure satisfies $T > R > P > S$. For a particular example, the donation game is a simplified form of the PD game, where a cooperator offers the other player a benefit $b$ at a cost $c$ to itself with $b > c$ but a defector offers nothing. The iterated Prisoner's Dilemma (IPD) games further assume repeated encounters between the same two players and sheds insights into the idea of direct reciprocity~\cite{trivers1971evolution}. 

Despite the astronomically huge individual behavioral strategy space for IPD game interactions~\cite{hilbe2018partners, harper2017reinforcement}, the so-called zero-determinant (ZD) strategies, including extortioners (extortionate ZD), compliers (generous ZD), and equalizers, are a set of rather simple memory-one strategies yet can unilaterally set a linear relation between their own payoff $s_X$ and their co-player's payoff $s_Y$~\cite{boerlijst1997equal, press2012iterated, stewart2013extortion, akin2015you, chen2022intricate}. The three free parameters that determine a ZD strategy $\bm{p} = [p_1, p_2, p_3, p_4]$ and how unequal the corresponding payoff relation $s_X - O = \chi(s_Y - O)$ is are the baseline payoff $O$, the extortion factor $\chi$, and the normalization factor $\phi$. An equalizer with $\chi = +\infty$ will fix their co-player's payoff to the given $O$, which can take any value between the punishment $P$ and the reward $R$. An extortioner with $O = P$ and $\chi > 1$ can always get an unfair share of the payoffs in conventional IPD games with $T + S > 2P$~\cite{press2012iterated}. In contrast, a complier with $O = R$ and $\chi > 1$ can guarantee that their own payoff is never above their co-player's~\cite{stewart2013extortion}.

As such, the finding of ZD strategies has greatly spurred new waves of work from diverse fields~\cite{hao2015extortion, mcavoy2016autocratic}, aiming to (i) elucidate the robustness and resilience of cooperation by means of the natural selection of IPD strategies from a population dynamics perspective~\cite{hilbe2013evolution, chen2014robustness, chen2022evolutionary} and (ii) explain reactions of human subjects to extortion that could be out of a desire for profit and/or a concern for fairness~\cite{hilbe2014extortion, becks2019extortion}. Despite the witting to gain the upper hand in IPD games, there still exists one open issue regarding ZD strategies: how will the payoff structure and the three parameters determine their pairwise dominance and extortion ability?

We have fully addressed the above question and revealed the unforeseen Achilles' heel of ZD strategies. Most of all, we have found and characterized multiple general classes of strategies that are unbending to extortion and can trigger the backfire of being extortionate and even outperform an extortioner in certain conditions. A fixed unbending strategy will compel a self-interested and adaptive extortioner who tries to maximize their payoff to be fair and behave like Tit-for-Tat (TFT)~\cite{axelrod1981evolution, nowak1992tit} ultimately by letting $\chi \to 1$. Examples of unbending strategies include general ZD strategies with the baseline payoff $O$ satisfying $P < O \leq R$ and the simple adaptive player Win-Stay Lose-Shift (WSLS)~\cite{nowak1993strategy}, of which the latter can even outperform an extortioner in interactions of more adversarial nature where $T + S < 2P$~\cite{d2015statistical}. 

Remarkably, the insistence of an unbending player on fairness can rein in not only extortioners but also a much broader family of strategies. In the present work, we further illustrate such a steering role toward altruism by considering the adaptive learning dynamics of a focal reactive player whose move in the current round depends on what the co-player did in the last round~\cite{nowak1992tit} against another fixed unbending co-player in a donation game~\cite{hilbe2013evolution}. In so doing, reactive players not only are de facto ZD strategies but also can be conveniently visualized within a unit square. It turns out that unbending strategies from different classes either ``train" the reactive player to behave like generous TFT~\cite{nowak1993strategy} or a full cooperator. Therefore, unbending players can helm the adaptive learning dynamics of extortionate and reactive players to fairness even in the absence of population dynamics and evolution. Our findings help illuminate ideas for multi-agent optimization in computational learning theory.

\section{Adaptive Learning Dynamics against Unbending Players} 

\paragraph{Payoff structure -- equal gains from switching}
As a simplified form of the PD game, the donation game features two parameters $b$ and $c$, representing the benefit and the cost of cooperation. The payoff structure follows $R = b - c$, $S = -c$, $T = b$, and $P = 0$. As a particular type of memory-one strategies, reactive strategies can be described by the vector $\bm{p} = [p_1, p_2, p_1, p_2]$, where $p_1$ and $p_2$ are the probabilities to cooperate after a cooperation or a defection by the co-player, respectively, and are known as the reactive norm. Reactive strategies in fact become ZD strategies under the setting of a donation game satisfying ``equal gains from switching", namely, $R + P = T + S = b - c$.

\begin{proposition}[Reactive strategies versus Zero-Determinant strategies]
In the donation game, reactive strategies is a subset of general Zero-Determinant strategies with the normalization factor $\phi = 1/(b\chi + 1)$ and the extortion factor $\chi$ being either positive or negative. 
\label{thm:reactive}
\end{proposition}
\begin{proof}
For any PD game, the set of ZD strategies is a collection of memory-one strategies $[p_1, p_2, p_3, p_4]$ with three free parameters $(O, \chi, \phi)$:
\begin{align}
\begin{split}
p_1 &= 1 - \phi(R - O)(\chi - 1), \\
p_2 &= 1 - \phi[(T - O)\chi + (O - S)], \\
p_3 &= \phi[(O - S)\chi + (T - O)], \\
p_4 &= \phi(O - P)(\chi - 1).
\end{split}
\end{align}
The four $p_i$'s describe the probability to cooperate after each pairwise outcome in \{CC, CD, DC, DD\}, respectively. 

In the donation game where the payoff matrix $[R, S, T, P]$ is replaced by $[b - c, -c, b, 0]$, the set of ZD strategies can be further written as
\begin{align}
\begin{split}
p_1 &= 1 - \phi(b - c - O)(\chi - 1), \\
p_2 &= 1 - \phi[(b - O)\chi + (O + c)], \\
p_3 &= \phi[(O + c)\chi + (b - O)], \\
p_4 &= \phi O(\chi - 1).
\end{split}
\end{align}
Given that the baseline payoff is between the punishment $P$ and the reward $R$ and that $0 \leq p_i \leq 1$ for $i = 1, 2, 3, 4$, the admissible ranges of $(O, \chi)$ are
\begin{equation}
\begin{gathered}
0 \leq O \leq b - c, \\
1 \leq \chi < +\infty \, \text{or} \,
-\infty < \chi \leq -\max\{\frac{O + c}{b - O}, \frac{b - O}{O + c}\}.
\end{gathered}
\end{equation}
Once these two parameters are decided, the range of $\phi$ can be derived accordingly ($\bm{p}$ needs to be a probability vector). If $\phi = 1/(b\chi + 1)$, it is straightforward to show that 
\begin{align}
\begin{split}
p_1 &= p_3 =  1 - \frac{(b - c - O)(\chi - 1)}{b\chi + c}, \\
p_2 &= p_4 = \frac{O(\chi - 1)}{b\chi + c},
\end{split}
\end{align}
and the linear relation $O(1 - p_1) - (b - c - O)p_2 = 0$ holds. That is, any reactive strategy $\bm{p} = [p_1, p_2, p_1, p_2]$ can be obtained by letting 
\begin{equation}
\begin{cases}
O = (b - c)p_2/(1 - p_1 + p_2), \\
\chi = [b - c(p_1 - p_2)]/[b(p_1 - p_2) - c], \\
\phi = 1/(b\chi + 1).
\end{cases}
\end{equation}
\end{proof}

Geometrically, the unit square defined by the two components $p_1$ and $p_2$ of reactive strategies contains the isosceles right triangle representing reactive ZD strategies with the extortion factor $\chi \geq 1$, whose size is determined by the benefit-to-cost ratio $r = b/c$. The lengths of both legs are in fact $1 - 1/r$. As is shown in Fig.~\ref{fig1}, the horizontal leg AB, the vertical leg AC, and the hypotenuse BC of the triangle ABC correspond to extortioners, compliers, and equalizers, respectively. More detailed illustrations are given in Table~\ref{tab:reactive_ZD}. 

\begin{figure}[htbp]
\centering
\includegraphics[width=0.45\textwidth]{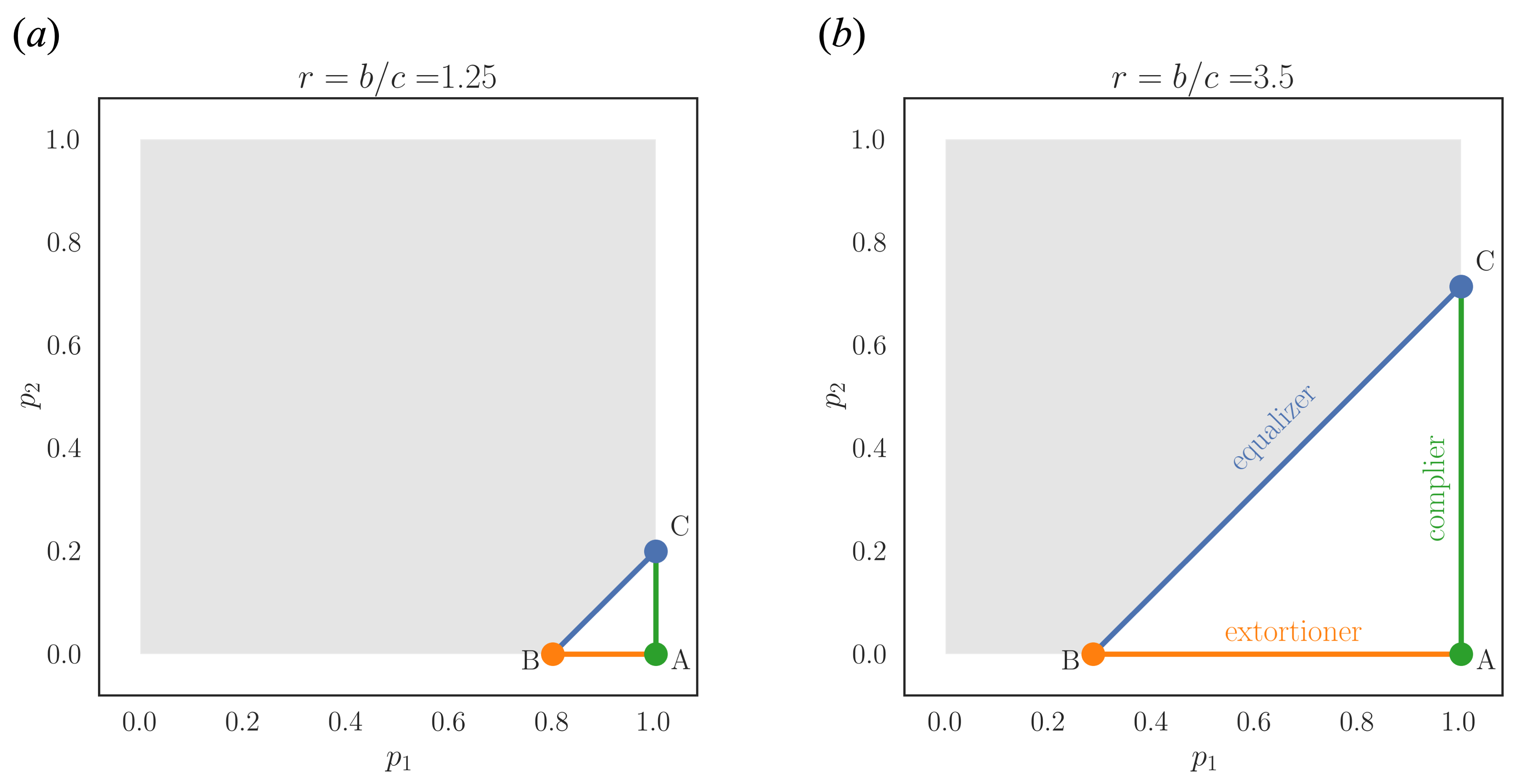}
\caption{Reactive strategies as a subset of general Zero-Determinant (ZD) strategies with the normalization factor $\phi = 1/(b\chi + 1)$ in the donation game. In each panel, the set of ZD strategies with a positive extortion factor $\chi$ is shaded in white (the isosceles right triangle), that with a negative $\chi$ in gray (the pentagon), and their union is the set of reactive strategies (the unit square). The three particular types of ZD strategies, known as extortioners, compliers, and equalizers, are highlighted in orange, green, and blue, respectively. }
\label{fig1}
\end{figure}

\begin{table}[htbp]
\caption{Examples of reactive strategies in the donation game that in fact have properties of Zero-Determinant (ZD) strategies.}
\centering
\tabulinesep=1.5mm
\begin{tabu}{| c | c | c |}
\hline
Type & ZD parameters &Reactive expressions \\
\hline
\multirow{2}{*}{extortioner} & $O = 0$ & $p_1 = 1 - \frac{(b - c)(\chi - 1)}{b\chi + c}$ \\
& $\chi > 1$ & $p_2 = 0$ \\
\hline
\multirow{2}{*}{complier} & $O = b - c$ & $p_1 = 1$ \\
& $\chi > 1$ & $p_2 = \frac{(b - c)(\chi - 1)}{b\chi + c}$ \\
\hline
\multirow{2}{*}{equalizer} & $0 \leq O \leq b - c$ & $p_1 = \frac{O + c}{b}$ \\
& $\chi = +\infty$ & $p_2 = \frac{O}{b}$ \\
\hline
\end{tabu}
\label{tab:reactive_ZD}
\end{table}

\paragraph{Steering control by unbending strategies}
As a routine, we assume that the focal player X uses strategy $\bm{p}$ and the co-player Y uses strategy $\bm{q}$. We then denote the average payoff of player X by $s_X$ and that of player Y by $s_Y$. Earlier on, we uncovered multiple general classes of fair-minded co-players who decide to fix their strategies such that a focal extortionate player can maximize their payoffs only if trying to be fair by letting $\chi \to 1$.  

\begin{definition}[Unbending strategies]
Let player X uses an extortionate Zero-Determinant strategy $\bm{p}$ of which the parameters satisfy $O = P$, $\chi > 1$, and $\phi > 0$. An unbending strategy $\bm{q}$ used by player Y is a memory-one strategy that (i) neutralizes the parameter $\phi$ in the first place such that both the two expected payoffs $s_X$ and $s_Y$ are independent of $\phi$, $\partial s_X/\partial \phi = \partial s_Y/\partial \phi = 0$, and (ii) guarantees that the derivative of $s_X$ with respect to $\chi$ is strictly negative, $\partial s_X/\partial \chi < 0$.
\label{defn:unbending}
\end{definition}

To scrutinize the behavior and hence the learning dynamics of a self-interested focal player from a much broader strategy space than those in previous studies (see Fig.~\ref{fig1}, the entire square instead of only those on the horizontal line AB), we now consider the scenario where $\bm{p}$ represents a reactive strategy used by an adapting player and $\bm{q}$ is a fixed unbending strategy from class A or class D.

\begin{example}[Class A of unbending strategies]
In the donation game, an unbending strategy $\bm{q} = [q_1, q_2, q_3, q_4]$ from class A can be described as: 
\begin{gather}
q_1 = 1, q_3 = 0, \\
q_a < q_2 < 1, 0 < q_4 \leq h_A(q_2).
\label{eq:unbending_A}
\end{gather}
The critical values for $q_2$ and $q_4$ are $q_a = (b - c)c/(b^2 + bc - c^2)$, $h_a$, $h_{Aa}$, and $h_{A}$ (for future reference as well):
\begin{equation}
\begin{cases}
h_a(q_2) = (bq_2 - c)(1 - q_2)/[b(1 - q_2) + c], \\
h_{Aa}(q_2) = (b - c)(1 - q_2)/c, \\
h_A(q_2) = (\ast)/(\ast\ast), \\
\end{cases}
\label{eq:h_A}
\end{equation}
where $(\ast) = (b - c)(1 - q_2)[(b^2 + bc - c^2)q_2 - (b - c)c]$ and $(\ast\ast) = bc^2q_2^2 - (b - c)(b^2 - bc - c^2)q_2 + (b - c)^2(b + c)$.
\label{ex:unbending_A}
\end{example}

\begin{example}[Class D of unbending strategies]
In the donation game, an unbending strategy $\bm{q} = [q_1, q_2, q_3, q_4]$ from class D can be described as: 
\begin{gather}
q_4 = h_D(q_1, q_2, q_3) = q_2 + q_3 - q_1, \\
\begin{aligned}
\begin{split}
d_D(q_1, q_2, q_3) = b(q_2 - q_1) + c(q_3 - q_1) + c &< 0, \\
q_2 + q_3 - q_1 &> 0.
\end{split}
\end{aligned}
\label{eq:h_D}
\end{gather}
Equivalently, class D of unbending strategies is the set of general Zero-Determinant (ZD) strategies with the baseline payoff $O$ satisfying $0 < O \leq b - c$. To be more specific, it is the relative complement of extortionate ZD strategies with respect to the set of general ZD strategies with the extortion factor $\chi$ satisfying $\chi > 1$, consisting of
\begin{enumerate}[(i)]
\item intermediate ZD: $0 < O < b - c$ and $q_1 < 1$,
\item generous ZD: $O = b - c$ and $q_1 = 1$. 
\end{enumerate}
\label{ex:unbending_D}
\end{example}

\begin{remark}
For any unbending strategy $\bm{q}$ from class A, $h_A(q_2) > h_a(q_2)$ and $h_{Aa}(q_2) > h_a(q_2)$ always hold. An intuitive comparison between $h_{A}$, $h_{Aa}$, and $h_a$ is given in Fig.~\ref{fig2}. For any unbending strategy $\bm{q}$ from class D, $q_3 > 0$ always holds.
\end{remark}

\begin{figure*}[htbp]
\centering
\includegraphics[width=0.8\textwidth]{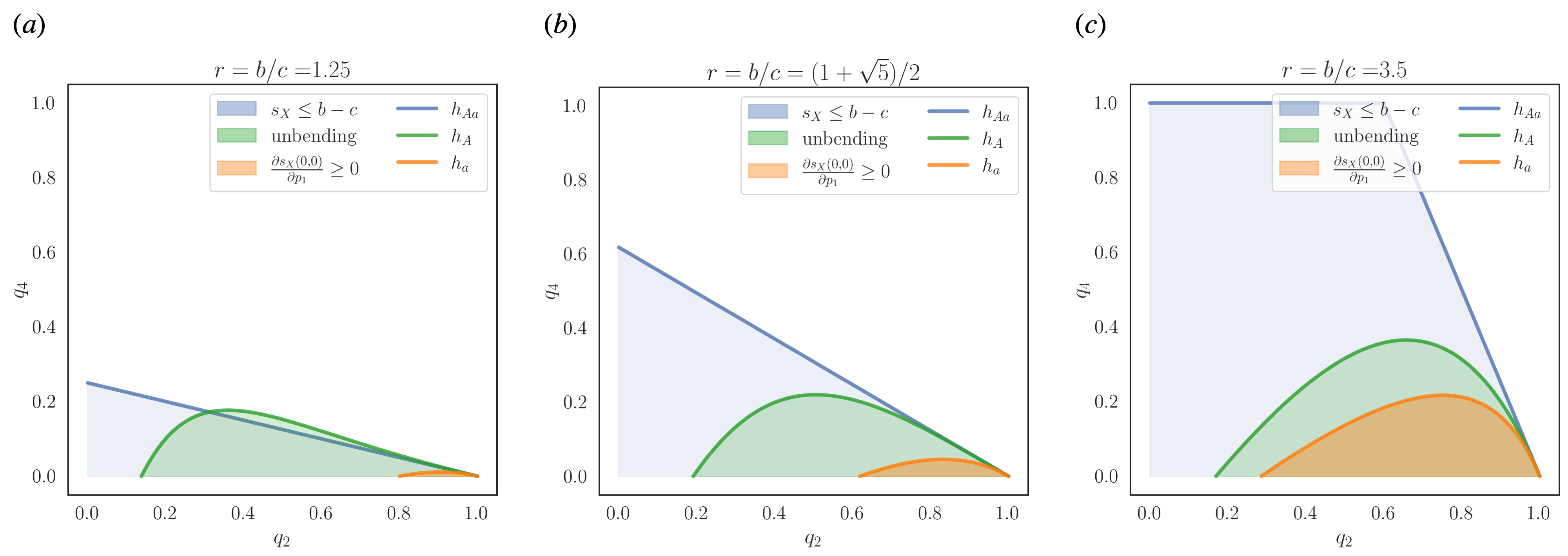}
\caption{Steering control by unbending strategies from class A. A fixed unbending player can exert unilateral influence on both the possible maximum payoff and the direction of strategy choices of any adaptive learning co-player. We find a golden ratio of the underlying payoff matrix, i.e., the benefit-to-cost ratio $r = b/c > (1 + \sqrt{5})/2$ that ensures the global maximum payoff at $(1, 0)$.}
\label{fig2}
\end{figure*}

As shown forthwith, we study how a fixed unbending player Y can curb the extortion and thus promote the cooperation of a self-interested, adaptive focal player X who explores the entire strategy space to get the highest possible payoff in the corresponding learning dynamics. Player X is assumed to use a reactive strategy within the parameter space $(p_1, p_2)$ which is broader than that of extortionate ZD strategies (see Table~\ref{tab:reactive_ZD}). To prove convergence results on possible learning outcomes, we investigate the monotonicity of the payoff $s_X$ with respect to both $p_1$ and $p_2$ and find where the maximum value can be achieved against unbending strategies from class A and class D, separately. We demonstrate that any evolutionary reactive player aiming to maximize their own payoff will be steered by an unbending co-player from extortion to fairness. For the sake of demonstration and ease of visualization, we have focused on the donation game in which reactive strategies are actually a subset of ZD strategies. It is straightforward to extend our results to more general IPD games.

\section{Steering Control Over Self-Interested Players}

In this section, we demonstrate the unilateral steering control by unbending players that use a fixed strategy from class A. Later on, we also show similar steering control by unbending players from class D. In order to visualize the learning dynamics, we choose to focus on the two-dimensional reactive strategies without loss of generality. Our analysis can be extrapolated for other memory-one strategies. In what follows, we explicitly prove the two major objectives of steering control over opponents, that is, characterize conditions for (i) yielding possible maximum payoffs and (ii) rendering fair and cooperative strategies, by using fixed unbending strategies in the adaptive learning dynamics.


\section*{the impact of unbending class a on steered learning dynamics} 

Let $\bm{p} = [p_1, p_2, p_1, p_2]$ be a reactive strategy and $\bm{q} = [1, q_2, 0, q_4]$ an unbending strategy from class A (as shown in Example~\ref{ex:unbending_A}). Using the method by Press and Dyson~\cite{press2012iterated}, we obtain the payoff $s_X(p_1, p_2)$ as a quadratic rational function of $p_1$ and $p_2$, resulting from the quotient of two determinants. If $p_1 = 1$, for example, $s_X(1, p_2) = s_Y(1, p_2) = b - c$. If $p_1 = p_2 = 0$, for another example, $s_X(0, 0) = bq_4/(1 - q_2 + q_4)$. In general, we have
\begin{equation}
s_X(p_1, p_2) - s_X(1, p_2) = \frac{(1 - p_1)\Delta(p_1, p_2)}{f_{A}(p_1, p_2)},
\label{eq:payoff_A}
\end{equation} 
where the denominator $f_{A}(p_1, p_2) = (1 - q_2)[(1 - p_1) - q_4(p_2 - p_1)^2] + q_4(1 - p_2)$ is always positive and
\begin{align}
\begin{split}
\Delta(p_1, p_2) &= (\ast) - (\ast\ast), \\
(\ast) &=  [bq_2q_4 + c(1 - q_2) - (b - c)q_4](1 - p_2), \\
(\ast\ast) &= b(1 - q_2)(q_4p_1 + 1 - q_4). \\
\end{split}
\label{eq:Delta}
\end{align}

\subsection{Maximum of the Payoff}

We claim that the maximal value of $s_X$ is obtained either at $p_1 = 1$  or $p_1 = p_2 = 0$. In addition to the proofs below, specific examples can be found in Table~\ref{tab:streamplot_A}. 

\begin{table*}[htbp]
\caption{Adaptive learning dynamics towards fairness and cooperation against unbending strategies from class A. The 2-dimensional stream plots (gradient vector field) from the adaptive learning dynamics of a self-interested reactive player X against a fixed unbending player Y are given with respect to $p_1$ and $p_2$. The color of the curves and dots $(p_1, p_2)$ corresponds to the payoff values of X in situ, as specified by the given color bar. The solid point on the corner denotes where the maximum value of $s_X$ lies (if there is more than one maximum point, only the bottom one with $p_2 = 0$ is shown). The other solid point indicates where $\partial s_X(p_1, 0)/\partial p_1$ becomes zero and the empty point is $(c/b, 0)$ (the left boundary of extortionate ZD strategies).}
\centering
\tabulinesep=1.5mm
\begin{tabu}{| c | c  c | c | c |}
\hline
& \multicolumn{2}{c|}{$r < r^{\ast}$} & $r = r^{\ast}$ & $r > r^{\ast}$ \\
&  \makecell[cc]{$\max s_X =$ \\ $s_X(0, 0)$} & \makecell[cc]{$\max s_X =$ \\ $s_X(1, p_1)$} 
& \makecell[cc]{$\max s_X =$ \\ $s_X(1, p_1)$} & \makecell[cc]{$\max s_X =$ \\ $s_X(1, p_1)$} \\
\hline
$\frac{\partial s_X(0, 0)}{\partial p_1} < 0$ & \raisebox{-.5\height}{\includegraphics[width=3.25cm]{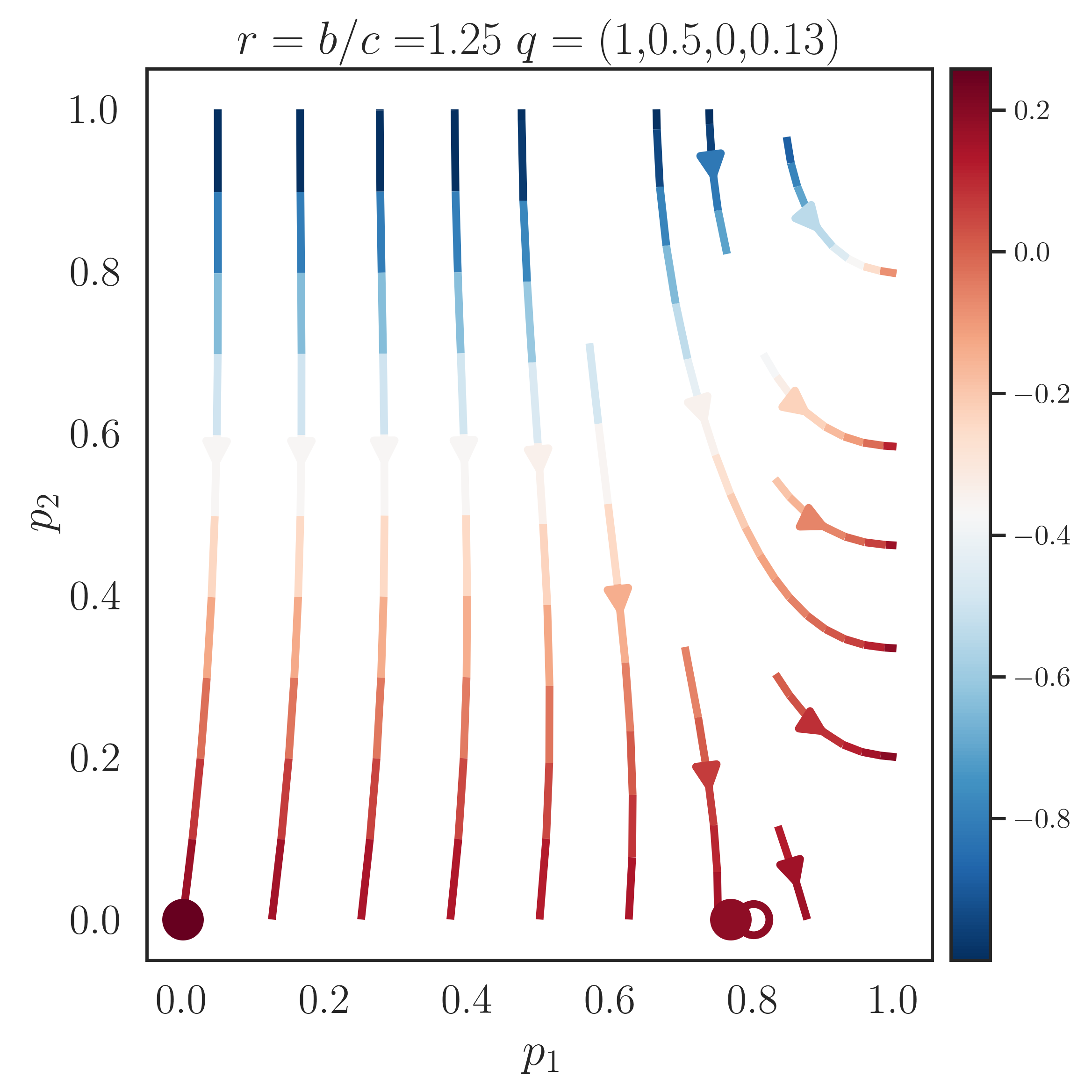}}
& \raisebox{-.5\height}{\includegraphics[width=3.25cm]{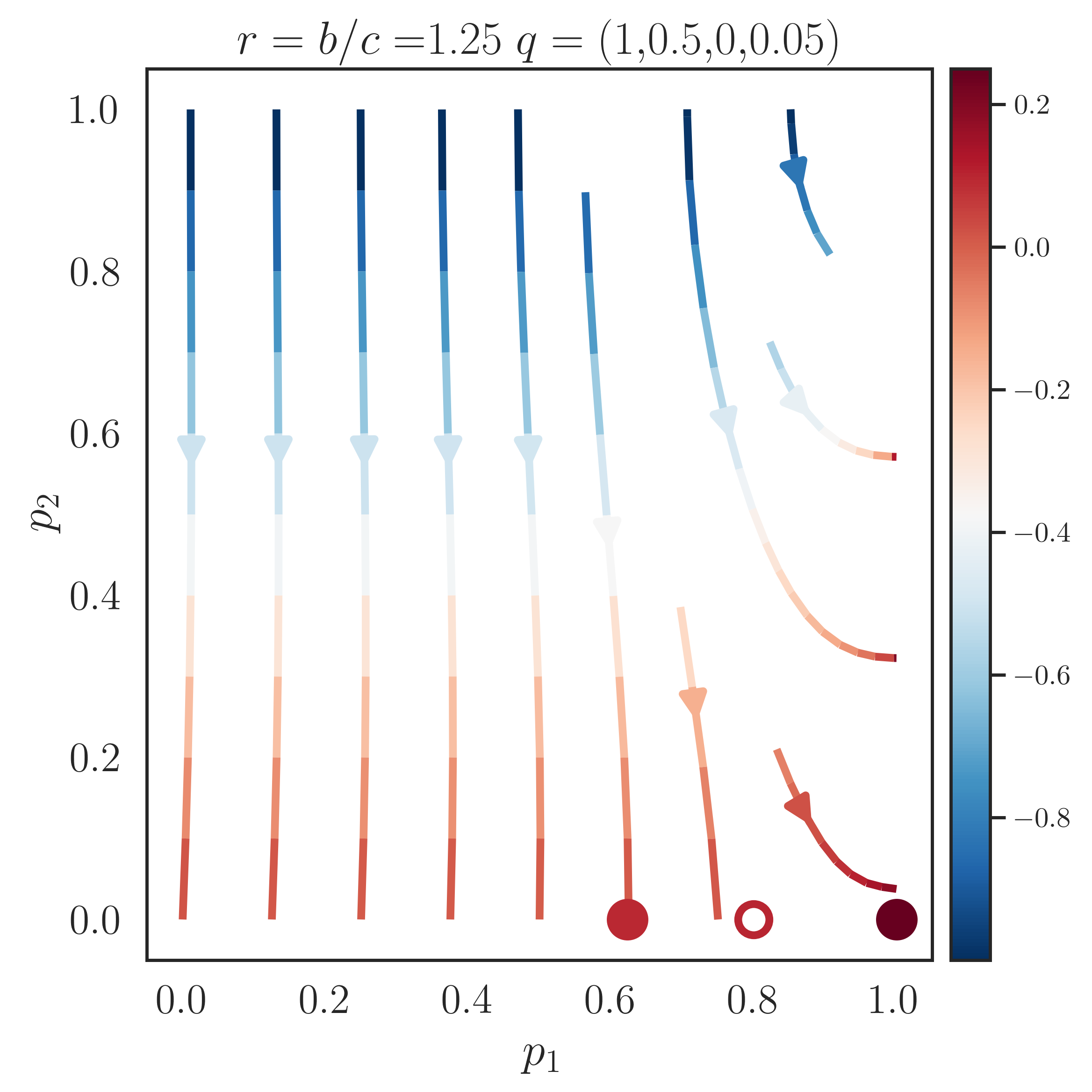}} 
& \raisebox{-.5\height}{\includegraphics[width=3.25cm]{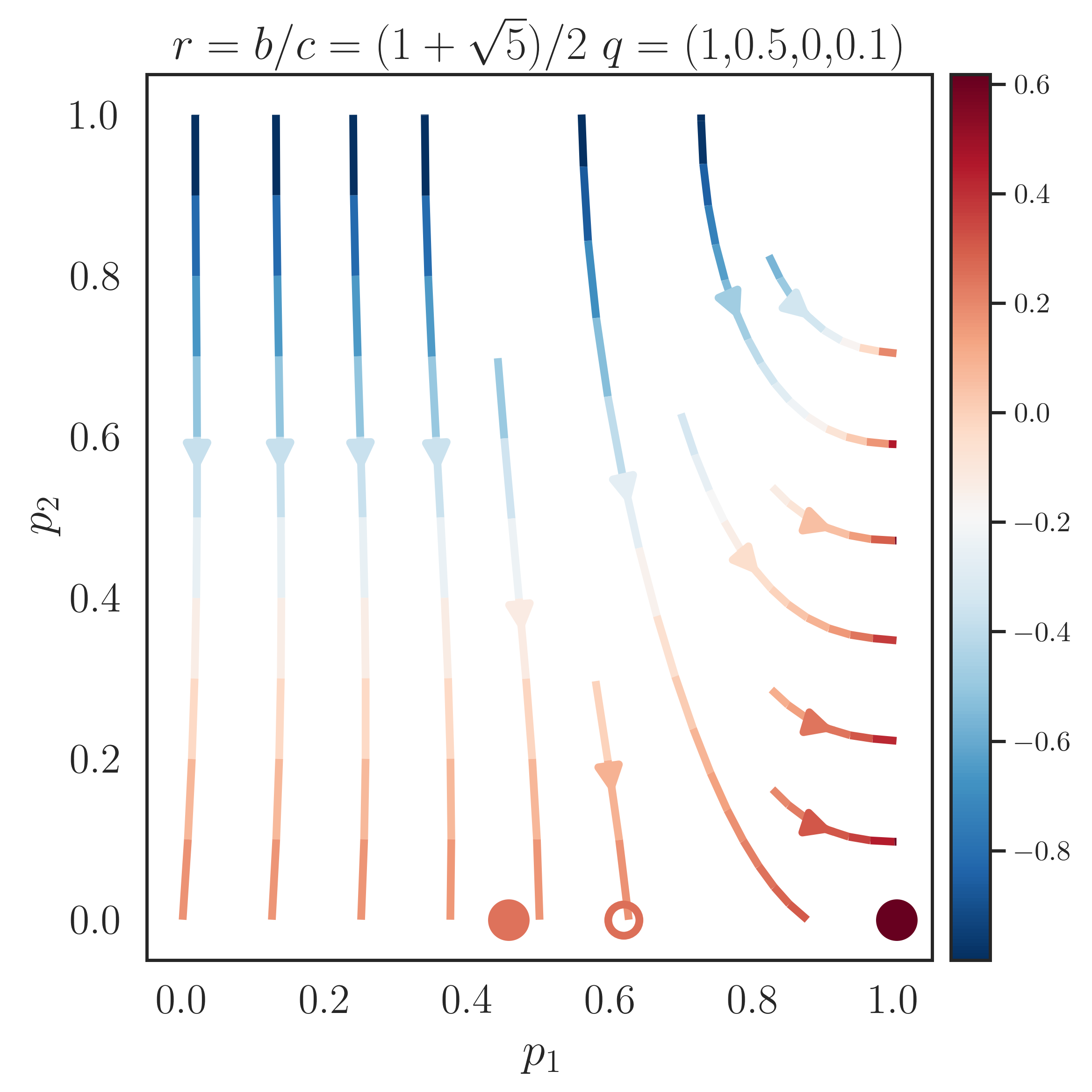}} 
&  \raisebox{-.5\height}{\includegraphics[width=3.25cm]{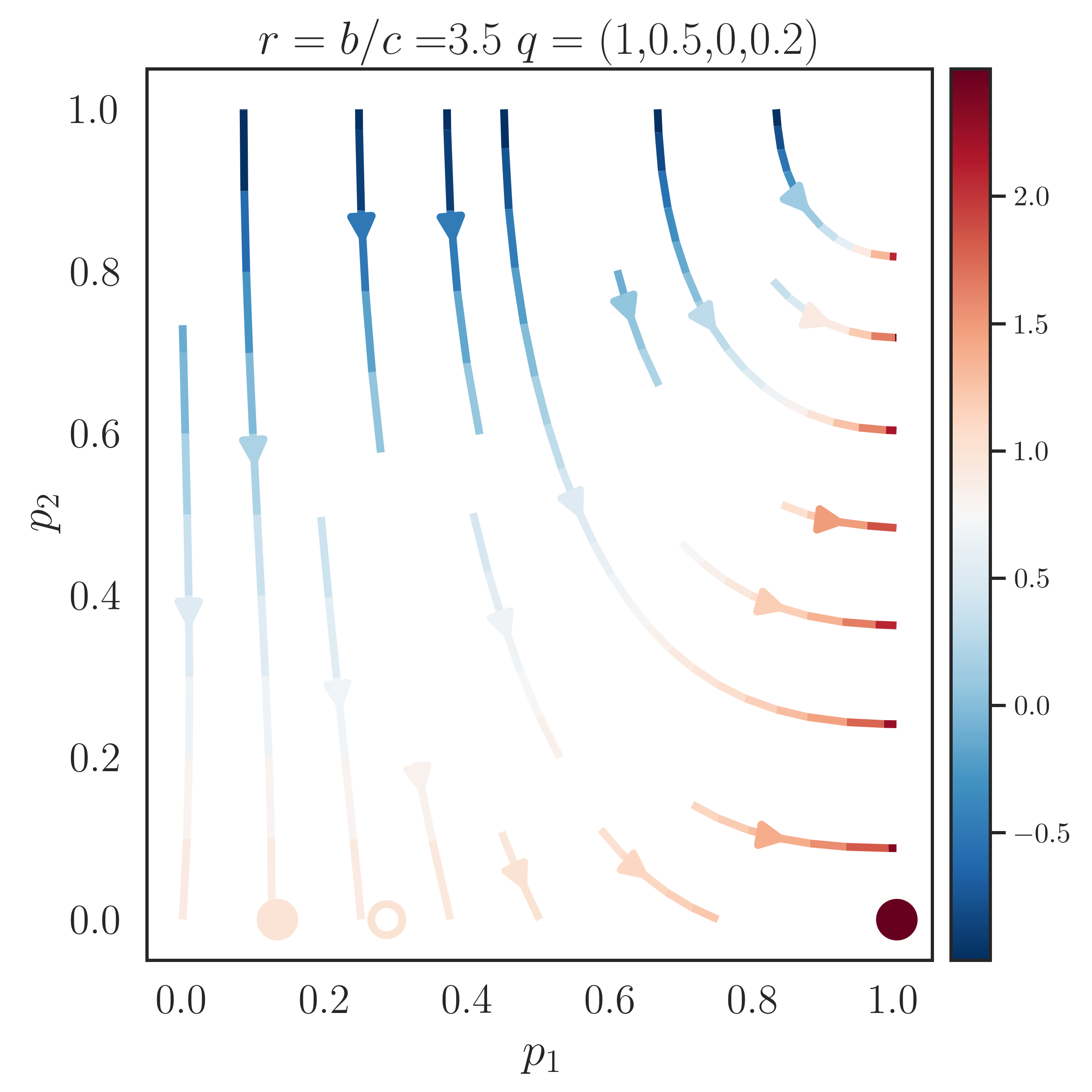}} \\
\hline
$\frac{\partial s_X(0, 0)}{\partial p_1} > 0$ & 
& \raisebox{-.5\height}{\includegraphics[width=3.25cm]{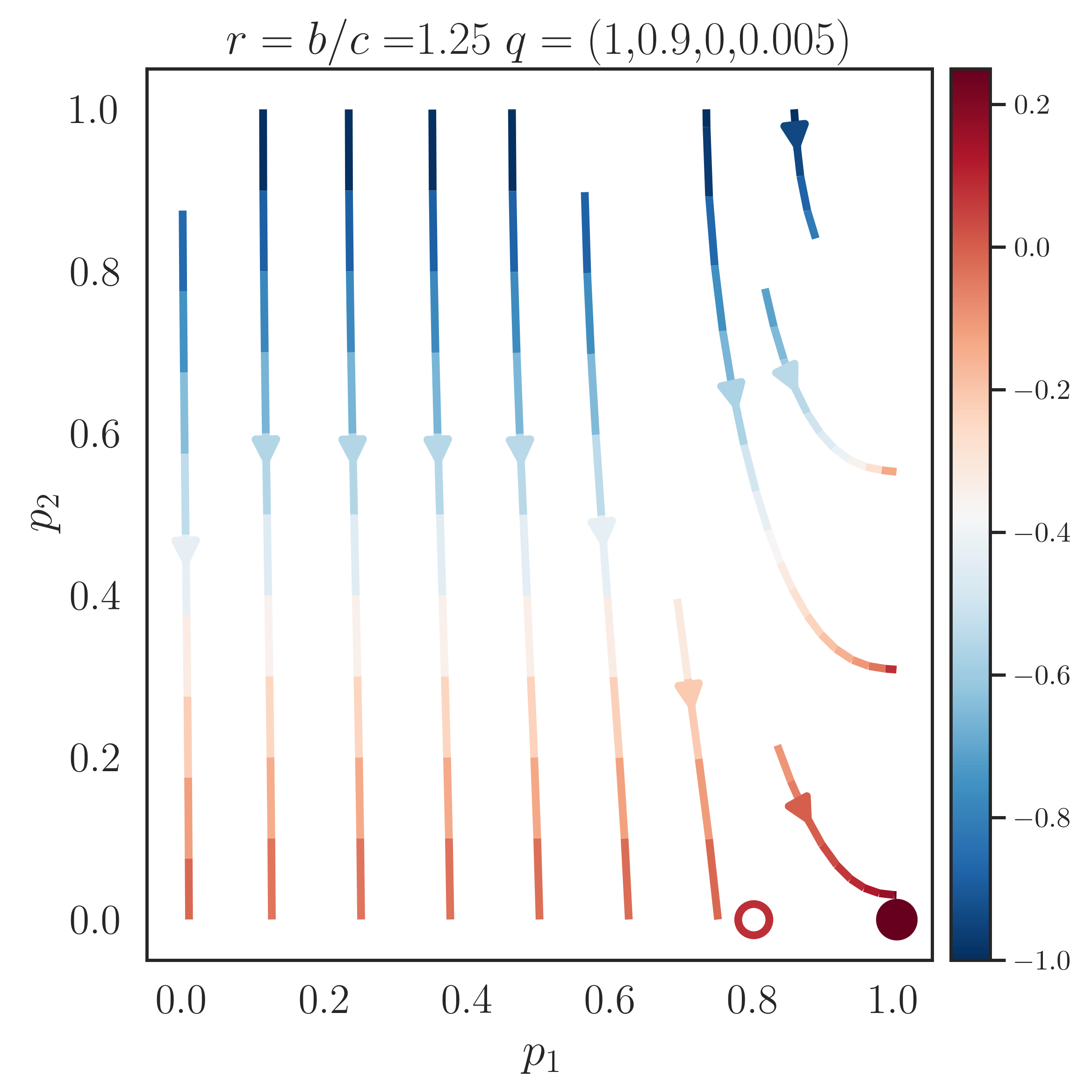}}
& \raisebox{-.5\height}{\includegraphics[width=3.25cm]{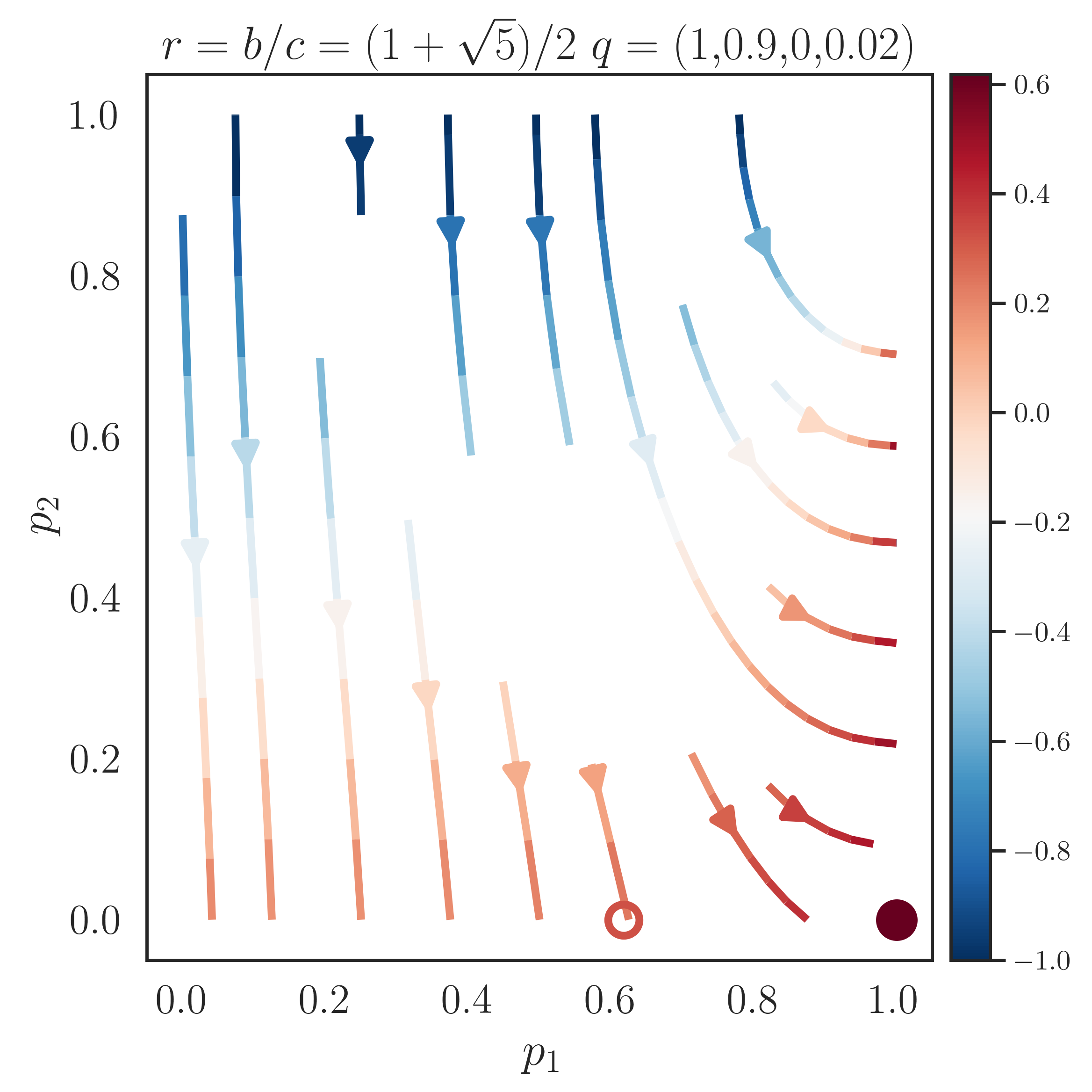}}
& \raisebox{-.5\height}{\includegraphics[width=3.25cm]{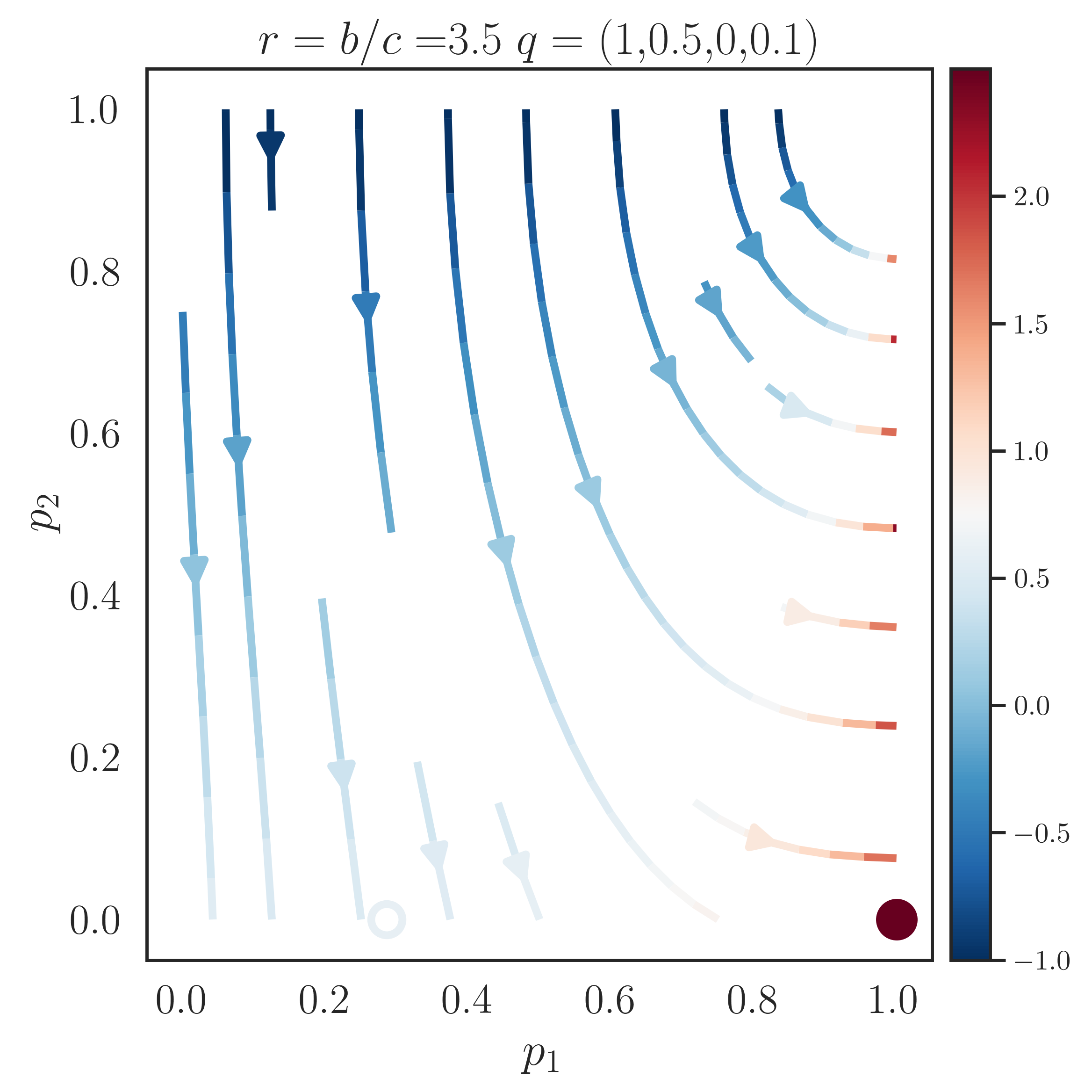}} \\
\hline
\end{tabu}
\label{tab:streamplot_A}
\end{table*}

\begin{theorem}[Maximum on the boundaries]
Assume that player X uses a reactive strategy $\bm{p}$ and that player Y uses an unbending strategy $\bm{q}$ from class A. The maximum of the expected payoff $s_X$ is
\begin{equation}
\max s_X(p_1, p_2) = 
\begin{cases}
s_X(1, p_2), & q_4 \leq h_{Aa}(q_2) \\
s_X(0, 0), & q_4 > h_{Aa}(q_2)
\end{cases}
\end{equation}
where $h_{Aa}(q_2)$ is given in~\eqref{eq:h_A}. 
\label{thm:maximum_A}
\end{theorem}
\begin{proof}
According to~\eqref{eq:Delta}, $\Delta(p_1, p_2)$ decreases with respect to both $p_1$ and $p_2$, whose maximal and minimal values are thus 
\begin{equation}
\begin{cases}
\Delta(0, 0) = c(q_4 - h_{Aa}), \\ 
\Delta(1, 1) = -b(1 - q_2),
\end{cases}
\end{equation}
respectively. Combining this observation with~\eqref{eq:payoff_A}, we decide that $p_1 = 1$ is a level curve with a local maximum. Moreover, the greatest payoff is achieved at $p_1 = 1$ if $\Delta(0, 0) \leq 0$ and $(0, 0)$ otherwise.
\end{proof}

We can further figure out whether the payoff $s_X$ always takes its maximum at $p_1 = 1$ by comparing $h_A$ and $h_{Aa}$ in ~\eqref{eq:h_A}.

\begin{corollary}[Rule of the golden ratio]
Assume that player X uses a reactive strategy $\bm{p}$ and that player Y uses an unbending strategy $\bm{q}$ from class A. The maximum of the expected payoff $s_X$ is always achieved at $p_1 = 1$ if the benefit-to-cost ratio $r \geq (1 + \sqrt{5})/2$ (the golden ratio).
\end{corollary}
\begin{proof}
The solutions to $h_{A}(q_2) = h_{Aa}(q_2)$ are $q_2 = 1$ and $q_2 = q_{A} = b(b - c)/c^2$. A routine calculation shows that $q_A > q_a$ and $h_{Aa}(q_A) > 0$. Therefore, whether $q_A \geq 1$ determines whether $\max s_X(p_1, p_2) = s_X(1, p_2)$ always holds, which further yields a comparison between the benefit-to-cost ratio $r$ and the golden ratio $(1 + \sqrt{5})/2$. 
\end{proof}

More specifically,
\begin{enumerate}[(i)]
\item if $1 < r < (1 + \sqrt{5})/2$, 
\begin{equation}
\max s_X(p_1, p_2) = 
\begin{cases}
s_X(1, p_2), & 0 < q_4 \leq h_{Aa} \\
s_X(0, 0), & h_{Aa} < q_4 \leq h_A
\end{cases}
\end{equation}
\item if $r \geq (1 + \sqrt{5})/2$, $\max s_X(p_1, p_2) = s_X(1, p_2)$.
\end{enumerate}

\subsection{Monotonicy of the Payoff}

We first take the partial derivative of the payoff $s_X$ with respect to $p_1$ and get
\begin{equation}
\frac{\partial s_X(p_1, p_2)}{\partial p_1} =\frac{ (1 - p_2)q_4g_A(p_1, p_2)}{f_A^2(p_1, p_2)},
\end{equation}
where $f_A(p_1, p_2)$ is defined as before and $g_A(p_1, p_2) = e_2p_1^2 + e_1p_1 + e_0$ is a quadratic function of $p_1$. We have
\begin{equation}
\begin{cases}
e_2 = -(1 - q_2)[b(1 - q_2) + c]q_4 - c(1 - q_2)^2, \\
e_1 = 2(1 - q_2)[b(1 - q_2)p_2 + b + c]q_4 + 2c(1 - q_2)^2,\\
\end{cases}
\end{equation} 
and 
\begin{align}
\begin{split}
e_2 + e_1 + e_0 &=  [p_2 + q_2(1 - p_2)][(\ast)q_4 + (\ast\ast)], \\
(\ast) &= [b(1 - q_2) - c](1 - p_2), \\
(\ast\ast) &= (1 - q_2)[b - c(1 - p_2)].
\end{split}
\end{align}

\begin{theorem}[Monotonicity along the first axis]
Assume that player X uses a reactive strategy $\bm{p}$ and that player Y uses an unbending strategy $\bm{q}$ from class A. Let $p_2$ be fixed. The expected payoff $s_X$ either always increases or first decreases and then increases with $p_1$.
\end{theorem}
\begin{proof}
Since $e_2 < 0$ and $2e_2 + e_1 = 2b(1 - q_2)q_4[p_2 + q_2(1 - p_2)] > 0$, the graph of $g_A$ as a function of $p_1$ is a parabola opening downward whose axis of symmetry is to the right of $p_1 = 1$. A routine calculation shows that $g_A(1, p_2) = e_2 + e_1 + e_0 > 0$. Therefore, as $p_1$ increases, $\partial s_X/\partial p_1$ is either always nonnegative or first negative and then positive for any fixed $p_2$.
\end{proof}

We can obtain more detailed results about $\partial s_X/\partial p_1$ on the boundaries of the unit square. 

\begin{corollary}[Monotonicity on the boundaries]
Assume that player X uses a reactive strategy $\bm{p}$ and that player Y uses an unbending strategy $\bm{q}$ from class A. Let $p_2 = 0$. The monotonicity of the expected payoff $s_X$ with respect to $p_1$ satisfies
\begin{enumerate}[(i)]
\item if $0 < q_4 \leq h_a(q_2)$, $\partial s_X/\partial p_1$ is always nonnegative,
\item if $h_a(q_2) < q_4 < h_A(q_2)$, $\partial s_X/\partial p_1$ is first negative and then positive,
\end{enumerate}
where $h_{a}(q_2)$ is given in~\eqref{eq:h_A}.
\end{corollary}
\begin{proof}
Based on the proposition above, it suffices to consider $\partial s_X/\partial p_1$ at $(0, 0)$, that is,
\begin{equation}
\left.\frac{\partial s_X(p_1, p_2)}{\partial p_1}\right|_{(0, 0)} 
= \frac{[b(1 - q_2) + c][h_a(q_2) - q_4]q_4}{(1 - q_2 + q_4)^2}.
\end{equation}
The sign of ${\partial s_X(p_1, p_2)}/{\partial p_1}$ at $(0, 0)$ is determined by the relation between $q_4$ and $h_a(q_2)$.
\end{proof}

\begin{remark}
In fact, $\partial s_X/\partial p_1$ is always positive at $(c/b, 0)$ and hence positive for any $c/b \leq p_1 \leq 1$. Namely, if player X uses a reactive extortionate ZD strategy and Y uses an unbending strategy from class A, the expected payoff $s_X$ will be irreverent to $p_2$ and an increasing function of $p_1$, which echoes with what we have found in previous studies. 
\end{remark}

\begin{example}[]
Assume that player X uses a reactive strategy $\bm{p}$ and that player Y uses an unbending strategy $\bm{q}$ from class A. Let 
\begin{equation}
\bm{q} = [1, (b - c)/b, 0, (2b - 3c)/4b].
\end{equation}
It is easy to tell that $(b - 2c)/2b = h_a < q_4 < h_A = (b - c)/2b$. The partial derivative of the expected payoff $s_X$ with respect to $p_1$ satisfies
\begin{equation}
\left.\frac{\partial s_X(p_1, p_2)}{\partial p_1}\right|_{p_2 = 0} 
\begin{cases}
< 0, & 0 \leq p_1 < c/(2b - c) \\
= 0, & p_1 = c/(2b - c) \\
> 0. & c/(2b - c) < p_1 \leq 1
\end{cases}
\end{equation}
\end{example}

We then take the partial derivative of the payoff $s_X$ with respect to $p_2$. The general expression of $\partial s_X/\partial p_2$ is long and complicated. We only enumerate the results on the boundaries.

\begin{proposition}[Monotonicity along the second axis]
Assume that player X uses a reactive strategy $\bm{p}$ and that player Y uses an unbending strategy $\bm{q}$ from class A. The monotonicity of the expected payoff $s_X$ with respect to $p_2$ satisfies
\begin{enumerate}[(i)]
\item when $p_1 = 0$, $\partial s_X/\partial p_2$ is always negative,
\item when $p_1 = 1$, $\partial s_X/\partial p_2$ is always zero,
\item when $p_2 = 1$ and $p_1 < 1$, $\partial s_X/\partial p_2$ is always negative.
\end{enumerate}
\end{proposition}

The detailed proof is omitted for the limit of space.

To summarize, the fixed unbending strategy $\bm{q}$ used by the co-player will determine the performance and the learning dynamics of the focal player. The global maximum payoff is always achieved at $(1, p_2)$ if $\bm{q}$ is taken from a subset of class A (Fig.~\ref{fig2}(a), $q_4 \leq h_{Aa}$) or even the entire set of class A of unbending strategies (Fig.~\ref{fig2}(b) and~\ref{fig2}(c), $q_4 < h_A$) provided that the benefit-to-cost ratio $r$ is at least the golden ratio $(1 + \sqrt{5})/2$. Furthermore, the learning dynamics can exhibit either global convergence (the shaded orange areas in Fig.~\ref{fig2}, $q_4 \leq h_a$) or bistability ($h_a < q_4 < h_A$) along the direction of change of $p_1$, where the final state of the focal player depends on their initial state. Examples are given in Table~\ref{tab:streamplot_A} and Fig.~\ref{fig3}.

\begin{figure*}[htbp]
\centering
\includegraphics[width=0.8\textwidth]{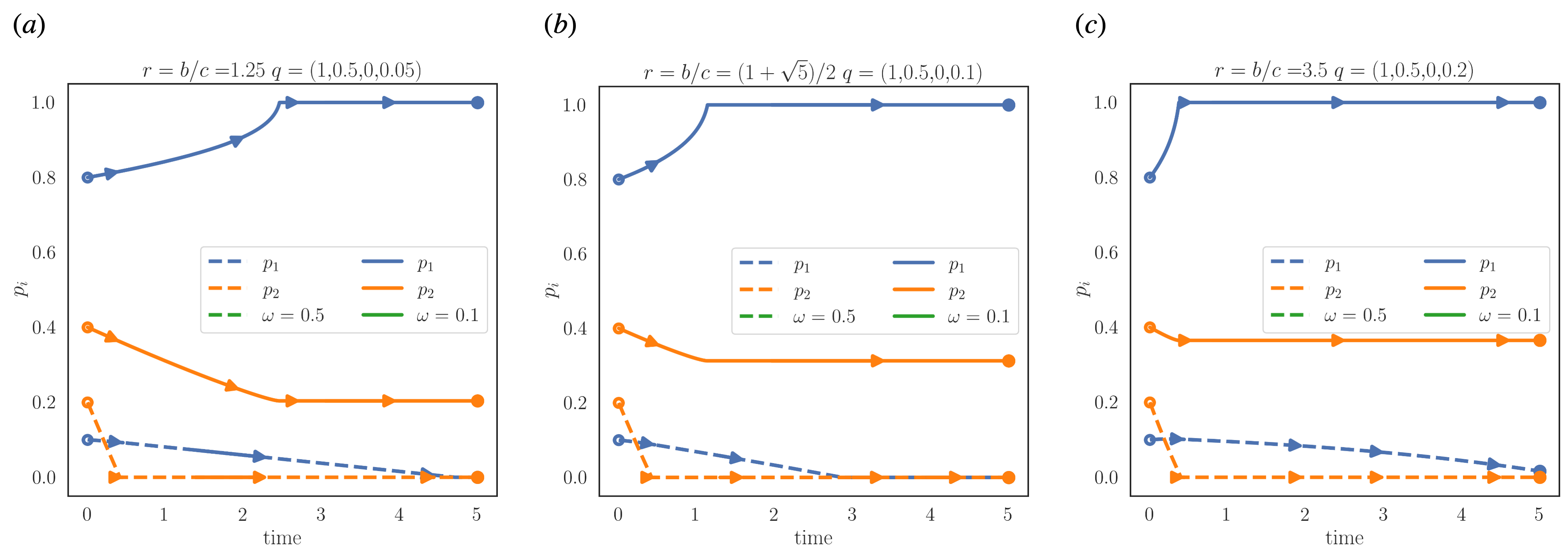}
\caption{Numerical learning processes enforced by unbending strategies from Class A. The learning curves of $p_1$ and $p_2$ of a reactive player are given with respect to time. The empty and the solid points represent the initial and the final states and the arrows indicate the direction. Different time scales are considered for different players illustrated by solid versus dashed curves: larger $\omega$ is for faster learning while smaller $\omega$ for slower learning.}
\label{fig3}
\end{figure*}

\section*{the impact of unbending class d on steered learning dynamics} 

Let $\bm{p} = [p_1, p_2, p_1, p_2]$ be a reactive strategy and $\bm{q} = [q_1, q_2, q_3, q_4]$ an unbending strategy from class D (as shown in Example~\ref{ex:unbending_D}). We obtain $s_X(p_1, p_2)$ as a linear rational function of $p_1$ and $p_2$. If $p_1 = p_2 = 1$, for example, $s_X(1, 1) = [(b - c)q_3 - c(1 - q_1)]/(1 - q_1 + q_3)$. If $p_1 = 1$ and $q_1 = 1$, for another example, the unbending co-player becomes a complier with $O = b - c$ and $s_X = s_Y = b - c$. In general, we have
\begin{equation}
s_X(p_1, p_2) - s_X(1, 1) = \frac{d_{D}(q_1, q_2, q_3)}{1 - q_1 + q_3} \cdot \frac{(\ast)}{(\ast\ast)}, 
\label{eq:payoff_D}
\end{equation}
where $(\ast) = q_3(1 - p_1) + (1 - q_1)(1 - p_2)$ and $(\ast\ast) = 1 - q_1 + q_3 - (q_2 - q_1)(p_2 - p_1)$.

\subsection{Maximum of the Payoff}

We claim that the maximal value of $s_X$ is obtained either at $p_1 = p_2 = 1$ or $p_1 = 1$.

\begin{theorem}[Maximum on the boundaries]
Assume that player X uses a reactive strategy $\bm{p}$ and that player Y uses an unbending strategy $\bm{q}$ from class D. The maximum of the expected payoff $s_X$ is 
\begin{equation}
\max s_X(p_1, p_2) = 
\begin{cases}
s_X(1, 1), & q_1 < 1 \\
s_X(1, p_2), & q_1 = 1
\end{cases}
\end{equation}
where the value of $q_1$ is determined by the baseline payoff $O$ of $\bm{q}$ as a general ZD strategy. 
\end{theorem}
\begin{proof}
The sign of the right-hand side in~\eqref{eq:payoff_D} determines whether $\max s_X(p_1, p_2) = s_X(1, 1)$. According to~\eqref{eq:h_D}, $-q_3 < q_2 - q_1 < 0$. Thus, $(\ast\ast) > 1 - q_1 + q_3 - q_3 \geq 0$. Therefore, we only need to consider the sign of $(\ast)$. It is straightforward to show that $(\ast) \geq 0$, which is an equality if and only if $p_1 = p_2 = 1$ or $p_1 = 1$ and $q_1 = 1$.
\end{proof}



%

\subsection{Monotonicity of the Payoff}

We take the partial derivatives of the payoff $s_X$ with respect to $p_1$ and $p_2$ and get
\begin{align}
\begin{split}
\frac{\partial s_X(p_1, p_2)}{\partial p_1} &= -\frac{[(q_2  - q_1)(1 - p_2) + q_3]d_{D}(q_1, q_2, q_3)}{(\ast\ast)^2}, \\
\frac{\partial s_X(p_1, p_2)}{\partial p_2} &= -\frac{[(q_2 - q_1)p_1 + 1 - q_2]d_{D}(q_1, q_2, q_3)}{(\ast\ast)^2}. \\
\end{split}
\end{align}
where $(\ast\ast)$ is the same as that in~\eqref{eq:payoff_D}. 

\begin{theorem}[Monotonicity across the unit square]
Assume that player X uses a reactive strategy $\bm{p}$ and that player Y uses an unbending strategy $\bm{q}$ from class D. The expected payoff $s_X$ is a strictly increasing function of $p_1$ and an increasing function of $p_2$:
\begin{align}
\begin{split}
\frac{\partial s_X(p_1, p_2)}{\partial p_1} &> 0 \\
\frac{\partial s_X(p_1, p_2)}{\partial p_2} &\geq 0.
\end{split}
\end{align}
The equality holds if and only if $p_1 = 1$ and $q_1 = 1$.
\end{theorem}
\begin{proof}
We know that $d_D < 0$ and $(\ast\ast) > 0$. Moreover,  $q_2 - q_1 > -q_3$ implies that $(q_2 - q_1)(1 - p_2) + q_3 > q_3p_2 \geq 0$. Hence, we have $\partial s_X/\partial p_1 > 0$. On the other hand, it is easy to tell that $(q_2 - q_1)p_1 + 1 - q_2 \geq 0$, which is an equality if and only if $p_1 = 1$ and $q_1 = 1$.
\end{proof}

To summarize,
\begin{enumerate}[(i)]
\item if Y uses an intermediate ZD strategy, 
\begin{equation}
\begin{gathered}
\max s_X(p_1, p_2) = s_X(1, 1) = \frac{(b - c)q_3 - c(1 - q_1)}{1 - q_1 + q_3}, \\
\frac{\partial s_X(p_1, p_2)}{\partial p_1} > 0 \, \text{and} \, \frac{\partial s_X(p_1, p_2)}{\partial p_2} > 0,
\end{gathered}
\end{equation}
\item if Y uses a generous ZD strategy,
\begin{equation}
\begin{gathered}
\max s_X(p_1, p_2) = s_X(1, p_2) = b - c, \\
\frac{\partial s_X(p_1, p_2)}{\partial p_1} > 0 \, \text{and} \, \frac{\partial s_X(p_1, p_2)}{\partial p_2} \begin{cases} > 0, & p_1 < 1 \\ = 0. & p_1 = 1\end{cases}
\end{gathered}
\end{equation}
\end{enumerate}

We have shown the role of unbending strategies in steering control and particularly their ability in enforcing fair and cooperative outcomes against payoff-maximizing players. Our further analysis reveals that some common IPD strategies are in fact unbending, including the PSO gambler (a machine-trained strategy)~\cite{harper2017reinforcement} and generous ZD strategies~\cite{stewart2013extortion}. In this regard, the framework of IPD has the potential for synergistically combining artificial intelligence (AI) and evolutionary game theory to enhance cooperation and foster fairness in various multi-agent systems. Also, all the unbending strategies from class A and at least those generous ZD strategies from class D are able to establish mutual cooperation among themselves. Moreover, even if noisy games are considered, for example, with probability $\varepsilon$ the intended move is implemented as the opposite one, their mutual cooperation is only impacted as $1 - \mathcal{O}(\varepsilon)$. These findings suggest the robustness and winning advantage of unbending strategies in population competition dynamics beyond the pairwise interactions considered above.

\section{Discussion \& Conclusion}

Our search for unbending strategies is directly motivated by suppressing extortion, thereby requiring targeted interactions with ZD opponents. Nevertheless, our work is more broadly motivated by how to foster and enforce fairness and cooperation in pairwise interactions in the IPD games. We reveal the previously unforeseen unbending properties of many well-known strategies in that they actually have the ability to resist extortion. For example, class A contains ``PSO Gambler'', an optimized memory-one strategy that is unbending when $T + S > 2P$. Other well-known examples of unbending strategies include WSLS from class A when $T + S < 2P$, ``willing'' from class C, and all the strategies from class D which are ZD players themselves but with a higher level of generosity than their opponents.

As such, unbending strategies can be used to rein in seemingly formidable extortionate ZD players, from whom a fair offer can ultimately be cultivated in their own interest. Our findings are in line with recent experimental evidence suggesting that human players often choose not to accede to extortion out of concern for fairness~\cite{hilbe2014extortion}. The conclusion can be further generalized and shed light on the backfire of extortion by unbending players in a broad sense. Our analysis shows that an adapting reactive player X after a higher payoff will cooperate more by increasing $p_1$ and in some cases $p_2$ to avoid potential punishment from a fixed unbending co-player Y in a donation game. Unbending strategies from class A would ``train" the reactive player to behave like generous TFT ($p_1 = 1$) whereas those from class D would even ``train" the reactive player to act as a full cooperator ($p_1 = p_2 = 1$). 

Under the influence of a co-player Y from class A of unbending strategies, the learning dynamics of a reactive player X may exhibit bistability along the direction of change of $p_1$, hence allowing two different learning outcomes, depending on both the original state of the focal player and the specific strategy of the co-player. After entering full cooperation with respect to $p_1$, player X will stay neutral along the direction of change of $p_2$. A reactive strategy can therefore converge to the full defector corner $(0, 0)$ or otherwise to the generous TFT edge $(1, p_2)$. Player X will always get the greatest payoff at the cooperative edge $(1, p_1)$ if the ratio $b/c >(\sqrt{5} + 1)/2$ (the ‘golden ratio’) or converge to the edge if player Y takes a strategy from a subset of class A.  On the other hand, unbending strategies from class D are able to guarantee global convergence to the same cooperative edge, after which the baseline payoff $O$ of player Y can further decide the direction of change of $p_2$. Player X will increase the value of $p_2$ until reaching the full cooperator corner $(1, 1)$ if player Y uses an intermediate ZD strategy with $P < O < R$ or remain neutral along the edge $(1, p_2)$ if player Y uses the generous ZD strategy with $O = R$.

In a nutshell, we have found and characterized general classes of unbending strategies that are able to steer the learning dynamics of selfish players in head-to-head encounters. Unbending strategies can trigger the backfire of greedy exploitations, punish the extortion, and thus turn the interactions into an Ultimatum game: to ``win each battle" or to ``win the war". These strategies helm the reactive learning dynamics of extortionate and reactive players to fairness without introducing population dynamics and evolution through generations. 

Our work helps pave the way for the promising initiative of combining game theory with artificial intelligence to gain more analytical insights into computational learning theory. Of particular interest are extensions to multi-agent learning systems that are fraught with perception and implementation errors and beyond pairwise interactions in a changing environment~\cite{bloembergen2015evolutionary}. In doing so, integrating theoretical and empirical approaches will help enhance our understanding of cooperation in various advanced AI systems besides social and biological systems~\cite{dafoe2021cooperative}.




\bibliographystyle{unsrt}
\bibliography{ref}

\end{document}